\documentclass[final,5p,times,twocolumn,numbers,sort&compress]{elsarticle}
\usepackage{amsmath,amssymb,amsthm,mathtools}
\usepackage{microtype}

\theoremstyle{plain}
\newtheorem{proposition}{Proposition}
\newtheorem{corollary}{Corollary}

\newcommand{\SPD}{\mathrm{Sym}^{+}_{n}}
\newcommand{\tr}{\operatorname{tr}}
\newcommand{\BW}{\mathrm{BW}}
\newcommand{\cL}{\mathcal{L}}
\newcommand{\cH}{\mathcal{H}}
\newcommand{\cR}{\mathcal{R}}
\newcommand{\grad}{\operatorname{grad}}

\journal{Physics Letters A}

\begin{document}

\begin{frontmatter}

\title{Hamiltonian Lift of Bures--Wasserstein Covariance Dynamics with a Spectral Floor}

\author[first]{Christian Kerskens\corref{cor1}}
\ead{kerskenc@tcd.ie}
\cortext[cor1]{Corresponding author}
\affiliation[first]{organization={Trinity College Institute of Neuroscience, Trinity College Dublin},
            city={Dublin},
            country={Ireland}}

\begin{abstract}
Covariance dynamics on the positive-definite cone are commonly described by gradient flows, which encode dissipative relaxation but obscure the underlying phase-space structure. We construct a finite-dimensional Hamiltonian lift of covariance dynamics on $\SPD$ equipped with the Bures--Wasserstein metric. The natural mechanical Lagrangian yields canonical momentum $\Pi=\tfrac12 L_\Sigma[\dot\Sigma]$, where $L_\Sigma$ is the Lyapunov operator, and explicit Hamiltonian $\cH(\Sigma,\Pi)=2\tr(\Pi\Sigma\Pi)+V(\Sigma)$. Adding Rayleigh dissipation recovers the Bures--Wasserstein gradient flow in the overdamped limit. For a spectral-floor and trace-control potential, the quadratic fluctuation Hamiltonian around the isotropic equilibrium separates trace and traceless modes; the baseline stiffness diverges as $(s-\nu)^{-2}$ as the equilibrium covariance approaches the floor. The construction identifies a conservative parent system for constrained Bures--Wasserstein covariance relaxation and fixes the local stiffness scale induced by the spectral floor.
\end{abstract}

\begin{highlights}
\item Bures--Wasserstein covariance flows are given a Hamiltonian lift.
\item The lift has explicit Hamiltonian $H=2\operatorname{tr}(\Pi\Sigma\Pi)+V$.
\item Rayleigh damping recovers the overdamped BW gradient flow.
\item A spectral floor induces stiffness scaling as $(s-\nu)^{-2}$.
\end{highlights}


\end{frontmatter}

\section{Introduction}
\label{sec:introduction}

Covariance matrices furnish a reduced state space for Gaussian probability measures, linear stochastic systems, continuous-variable quantum systems, and near-equilibrium fluctuation theories. The cone $\SPD$ of positive-definite covariance matrices carries several natural Riemannian geometries; the Bures--Wasserstein metric is distinguished by its relation to optimal transport between centred Gaussian measures and to the Bures metric in quantum information geometry \citep{givens1984class,takatsu2011,villani2009optimal,malago_pistone,bhatia2019bures,bengtsson2017geometry}.

Many dynamical uses of this geometry are dissipative. A free-energy functional $F(\Sigma)$ generates a gradient flow $\dot\Sigma=-\grad_{\BW}F(\Sigma)$, natural in thermodynamics, stochastic processes, and optimal-transport formulations of dissipation \citep{jko1998,otto2001,seifert2012stochastic,ambrosio2008gradient}. Such first-order equations do not display a phase-space structure: canonical momenta, symplectic form, inertial corrections, and quadratic fluctuation Hamiltonians are absent from the gradient-flow description itself. This is a structural limitation when one wants to analyse inertial corrections, linear response, or fluctuation spectra associated with covariance relaxation.

Lower spectral bounds also arise naturally in covariance regularization. In signal processing and high-dimensional statistics, eigenvalue flooring, diagonal loading, and shrinkage are used to maintain invertibility, limit condition numbers, and stabilize inverse-covariance operations in applications such as adaptive filtering, beamforming, and Gaussian classification \citep{vantrees2002optimum,ledoit2004well,chen2010shrinkage}. Trace penalties or constraints can additionally regulate total variance or signal power. The potential considered below should be viewed as a smooth dynamical analogue of these ideas: its logarithmic term enforces an eigenvalue floor, while its trace-dependent term controls the overall covariance scale. We do not claim that this precise potential is a standard covariance estimator; rather, it provides an analytically tractable example for examining the mechanics induced by a spectral boundary.

This Letter constructs the natural mechanical system associated with the Bures--Wasserstein metric and derives its cotangent-bundle Hamiltonian. The main point is that the Bures--Wasserstein metric gives not only a distance and a gradient operator, but also an explicit kinetic Hamiltonian on the covariance cotangent bundle. Rayleigh damping reduces the lift to the usual gradient flow in the overdamped limit. We then introduce a spectral-floor and trace-control potential, compute the isotropic equilibrium, and derive the quadratic fluctuation Hamiltonian. The Hessian separates trace and traceless covariance modes, and exhibits a boundary stiffness scaling proportional to $(s-\nu)^{-2}$ as the covariance floor is approached.

\section{Bures--Wasserstein covariance mechanics}
\label{sec:bw_mechanics}

For $\Sigma\in\SPD$, tangent vectors at $\Sigma$ are identified with symmetric matrices. The Lyapunov operator $L_\Sigma[A]$ \citep{bhatia2009positive,horn1991topics} assigns to each symmetric $A$ the unique symmetric solution of
\begin{equation}
\Sigma L_\Sigma[A]+L_\Sigma[A]\Sigma=A.
\label{eq:lyapunov_operator}
\end{equation}
The Bures--Wasserstein metric on $\SPD$ is
\begin{equation}
g_{\BW,\Sigma}(A,B)=\tfrac12\tr\!\left(L_\Sigma[A]B\right),
\label{eq:bw_metric}
\end{equation}
and we consider the natural mechanical Lagrangian \citep{arnold1989mathematical}
\begin{equation}
\cL(\Sigma,\dot\Sigma)=\tfrac12 g_{\BW,\Sigma}(\dot\Sigma,\dot\Sigma)-V(\Sigma),
\label{eq:natural_lagrangian}
\end{equation}
with $V$ a smooth potential on an open subset of $\SPD$.

\begin{proposition}[Hamiltonian lift]
\label{prop:hamiltonian_lift}
The Lagrangian \eqref{eq:natural_lagrangian} has canonical momentum $\Pi=\tfrac12 L_\Sigma[\dot\Sigma]$ and Hamiltonian
\begin{equation}
\cH(\Sigma,\Pi)=2\tr(\Pi\Sigma\Pi)+V(\Sigma),
\label{eq:hamiltonian}
\end{equation}
with equations of motion
\begin{equation}
\dot\Sigma=2(\Sigma\Pi+\Pi\Sigma),\qquad
\dot\Pi=-2\Pi^2-\frac{\partial V}{\partial\Sigma}.
\label{eq:hamilton}
\end{equation}
\end{proposition}

\begin{proof}
The kinetic Lagrangian is $\cL_{\mathrm{kin}}=\tfrac14\tr(L_\Sigma[\dot\Sigma]\dot\Sigma)$, so, using the trace pairing to identify symmetric matrices with their duals, $\Pi=\partial\cL/\partial\dot\Sigma=\tfrac12 L_\Sigma[\dot\Sigma]$. The Lyapunov equation \eqref{eq:lyapunov_operator} with $A=\dot\Sigma$ then gives $\dot\Sigma=2(\Sigma\Pi+\Pi\Sigma)$. Computing $\cH_{\mathrm{kin}}=\tr(\Pi\dot\Sigma)-\cL_{\mathrm{kin}}=2\tr(\Pi\Sigma\Pi)$, and differentiating \eqref{eq:hamiltonian} with respect to $\Pi$ and $\Sigma$, recovers \eqref{eq:hamilton}.
\end{proof}

The explicit kinetic Hamiltonian \eqref{eq:hamiltonian} is specific to the Lyapunov representation of the Bures--Wasserstein metric. In a generic Riemannian formulation the kinetic Hamiltonian is written abstractly using the inverse metric; here the inverse metric is realised by the matrix map $\Pi\mapsto 2(\Sigma\Pi+\Pi\Sigma)$, giving the closed expression $2\tr(\Pi\Sigma\Pi)$. This form is what makes the subsequent fluctuation calculation elementary.

For a regular solution, the kinematic equation has Lyapunov form $\dot\Sigma=\Sigma X+X\Sigma$ with $X=2\Pi$. Hence positivity is preserved as long as the solution remains finite, since $\Sigma(t)$ evolves by congruence. The phase space carries the canonical symplectic form $\omega=\tr(d\Pi\wedge d\Sigma)$, preserved by the Hamiltonian flow of \eqref{eq:hamiltonian} by Liouville's theorem.

As a scalar check, when $n=1$ the Lyapunov equation gives $L_\Sigma[A]=A/(2\Sigma)$ and hence $g_{\BW,\Sigma}(A,A)=A^2/(4\Sigma)$. Thus $\cL_{\rm kin}=\dot\Sigma^2/(8\Sigma)$, $\Pi=\dot\Sigma/(4\Sigma)$, and $\cH_{\rm kin}=2\Sigma\Pi^2$, exactly the one-dimensional reduction of \eqref{eq:hamiltonian}.

\subsection{Trace and determinant dynamics}
\label{sec:trace_determinant}

The Hamiltonian equations also give simple evolution identities for two basic covariance invariants. Taking the trace of the kinematic equation yields
\begin{equation}
\frac{d}{dt}\tr\Sigma=4\tr(\Sigma\Pi).
\label{eq:trace_dynamics}
\end{equation}
Thus the total covariance scale is generally not conserved, but is coupled to the momentum through $\tr(\Sigma\Pi)$.

Jacobi's formula gives
\begin{align}
\frac{d}{dt}\log\det\Sigma
&=\tr(\Sigma^{-1}\dot\Sigma) \nonumber\\
&=2\tr\!\left(\Pi+\Sigma^{-1}\Pi\Sigma\right)
=4\tr\Pi,
\label{eq:logdet_dynamics}
\end{align}
and consequently
\begin{equation}
\frac{d}{dt}\det\Sigma=4\det(\Sigma)\tr\Pi.
\label{eq:det_dynamics}
\end{equation}
The determinant is therefore constant on any time interval on which the momentum remains traceless. From the second Hamilton equation,
\begin{equation}
\frac{d}{dt}\tr\Pi=-2\tr(\Pi^2)-\tr\!\left(\frac{\partial V}{\partial\Sigma}\right).
\label{eq:momentum_trace_dynamics}
\end{equation}
These scalar equations are not closed in general, since they retain information about the full matrix variables. For the potential introduced in Section~\ref{sec:potential}, the isotropic sector $\Sigma=\sigma I$, $\Pi=pI$ is invariant and reduces to
\begin{equation}
\dot\sigma=4\sigma p,\qquad
\dot p=-2p^2+\frac{\kappa}{2(\sigma-\nu)}-\alpha(\sigma^3-\lambda).
\label{eq:isotropic_dynamics}
\end{equation}
This sector makes explicit how the barrier prevents the covariance scale from reaching the floor and how the trace-control term determines the restoring force.

\section{Overdamped reduction}
\label{sec:overdamped}

Adding Rayleigh dissipation \citep{goldstein2002classical} $\cR(\Sigma,\dot\Sigma)=\tfrac\gamma2 g_{\BW,\Sigma}(\dot\Sigma,\dot\Sigma)$ with $\gamma>0$, the dissipative Euler--Lagrange equation reads, after writing $F=V/\gamma$ and taking $\gamma\to\infty$ with $F$ fixed,
\begin{equation}
\frac{\partial F}{\partial\Sigma}=-\tfrac12 L_\Sigma[\dot\Sigma].
\label{eq:overdamped_balance}
\end{equation}
The defining relation $\partial F/\partial\Sigma=\tfrac12 L_\Sigma[\grad_{\BW}F]$ of the Bures--Wasserstein gradient combined with the invertibility of $L_\Sigma$ then yields
\begin{equation}
\dot\Sigma=-\grad_{\BW}F(\Sigma).
\label{eq:bw_gradient_flow}
\end{equation}
The Bures--Wasserstein gradient flow is therefore the overdamped Rayleigh reduction of the Hamiltonian system \eqref{eq:hamilton}, not the Legendre transform of it.

\section{Spectral-floor trace potential}
\label{sec:potential}

Let $\Phi(\Sigma)=n^{-1}\tr\Sigma$. For parameters $\nu\geq 0$, $\kappa>0$, $\alpha>0$, $\lambda>0$, and domain $\mathcal D_\nu=\{\Sigma\in\SPD:\Sigma\succ\nu I\}$, define
\begin{equation}
V(\Sigma)=-\tfrac\kappa2\log\det(\Sigma-\nu I)
+\alpha n\!\left(\tfrac14\Phi(\Sigma)^4-\lambda\Phi(\Sigma)\right).
\label{eq:potential}
\end{equation}
The logarithmic barrier imposes a spectral floor at $\nu$ \citep{boyd2004convex}, while the trace term controls the total covariance scale through the preferred value $\lambda$; here $\lambda$ has the dimension of $\Phi^3$. The first variation is
\begin{equation}
\frac{\partial V}{\partial\Sigma}=-\tfrac\kappa2(\Sigma-\nu I)^{-1}+\alpha(\Phi^3-\lambda)I.
\label{eq:potential_gradient}
\end{equation}
For an isotropic equilibrium $\Sigma_0=sI$ with $s>\nu$, the stationarity condition becomes
\begin{equation}
\alpha(s^3-\lambda)=\frac{\kappa}{2(s-\nu)}.
\label{eq:equilibrium}
\end{equation}
The spectral barrier shifts the equilibrium away from the uncoupled trace value $s^3=\lambda$; in the far-from-boundary regime the shift is perturbatively small.

\section{Quadratic fluctuations and boundary stiffness}
\label{sec:quadratic}

Let $\Sigma=sI+Q$ and $\Pi=P$ with symmetric fluctuations $Q,P$. Since the equilibrium momentum is zero, the quadratic Hamiltonian is
\begin{equation}
\cH^{(2)}=2s\tr(P^2)+\tfrac12\,Q\!:\!\mathcal K\!:\!Q,
\label{eq:quadratic_hamiltonian}
\end{equation}
with $\mathcal K$ the Hessian of $V$ at $sI$.

\begin{proposition}[Hessian at the isotropic equilibrium]
\label{prop:hessian}
At $\Sigma_0=sI$,
\begin{equation}
K_{ij,kl}=\tfrac{C}{2}\!\left(\delta_{ik}\delta_{jl}+\delta_{il}\delta_{jk}\right)+\frac{3\alpha s^2}{n}\delta_{ij}\delta_{kl},
\qquad
C=\frac{\kappa}{2(s-\nu)^2}.
\label{eq:hessian}
\end{equation}
\end{proposition}

\begin{proof}
From $\delta(\Sigma-\nu I)^{-1}=-(\Sigma-\nu I)^{-1}(\delta\Sigma)(\Sigma-\nu I)^{-1}$, the barrier contribution to the Hessian at $sI$ is $C\cdot\tfrac12(\delta_{ik}\delta_{jl}+\delta_{il}\delta_{jk})$ with $C$ as in \eqref{eq:hessian}. For the trace term, $\delta[\alpha(\Phi^3-\lambda)I]=(3\alpha\Phi^2/n)\tr(\delta\Sigma)\,I$, which at $\Phi=s$ gives the second contribution.
\end{proof}

Decomposing $Q=q_0\,I/\sqrt n+Q_{\mathrm{tl}}$ with $\tr Q_{\mathrm{tl}}=0$,
\begin{equation}
\tfrac12\,Q\!:\!\mathcal K\!:\!Q
=\tfrac12\!\left(C+3\alpha s^2\right)q_0^2
+\tfrac12 C\tr(Q_{\mathrm{tl}}^2).
\label{eq:trace_traceless_energy}
\end{equation}
The trace mode receives the additional stiffness $3\alpha s^2$; traceless modes feel only the geometric stiffness $C$.

The linearized Hamilton equations are
\begin{equation}
\dot Q=4sP,\qquad \dot P=-\mathcal K[Q],
\label{eq:linearized_hamilton}
\end{equation}
and hence
\begin{equation}
\ddot Q=-4s\,\mathcal K[Q].
\label{eq:linearized_second_order}
\end{equation}
The squared frequencies of the traceless and normalized trace modes are therefore
\begin{equation}
\omega_{\mathrm{tl}}^2=4sC,\qquad
\omega_{\mathrm{tr}}^2=4s(C+3\alpha s^2).
\label{eq:mode_frequencies}
\end{equation}
Thus the stiffness diverges as $(s-\nu)^{-2}$, while the corresponding linear oscillation frequencies diverge as $(s-\nu)^{-1}$ near the spectral floor.

\begin{corollary}[Boundary stiffness scaling]
\label{cor:boundary_stiffness}
As $s\to\nu^+$,
\begin{equation}
C(s,\nu)=\frac{\kappa}{2(s-\nu)^2}\to\infty,
\label{eq:boundary_scaling}
\end{equation}
so all local covariance fluctuation modes acquire an unbounded restoring stiffness as the spectral floor is approached, while the trace-control contribution $3\alpha s^2$ remains finite for finite~$s$.
\end{corollary}

This is a local statement about the quadratic Hamiltonian near the constrained isotropic state. On the open domain $\mathcal D_\nu$ the Hamiltonian vector field is smooth, and finite-energy trajectories cannot cross the logarithmic barrier because $V\to+\infty$ as the smallest eigenvalue of $\Sigma-\nu I$ tends to zero.

\section{Discussion}
\label{sec:discussion}

The Bures--Wasserstein gradient flow admits a Hamiltonian lift to a finite-dimensional cotangent bundle with canonical symplectic structure. The gradient flow is not the Legendre transform of the lift but its overdamped Rayleigh reduction, so the Hamiltonian system \eqref{eq:hamilton} acts as a conservative parent theory for dissipative Bures--Wasserstein covariance relaxation rather than a replacement for it.

For the spectral-floor trace potential \eqref{eq:potential}, the Hessian at the isotropic equilibrium separates trace and traceless modes, with baseline stiffness $C=\kappa/[2(s-\nu)^2]$ diverging as the equilibrium covariance approaches the floor. The quadratic Hamiltonian \eqref{eq:quadratic_hamiltonian} is a Gaussian Hamiltonian on the $N=n(n+1)/2$-dimensional symmetric-matrix phase space and generates an element of $\mathfrak{sp}(2N,\mathbb R)$ \citep{arvind1995real}. In oscillator variables, any such quadratic Hamiltonian admits the usual decomposition into anomalous and number-preserving terms \citep{adesso2014}. These are represented, on two-mode subblocks, by the standard $\mathfrak{su}(1,1)$ squeezing and $\mathfrak{su}(2)$ rotation algebras of Gaussian dynamics \citep{yurke1986su2}. A detailed sector classification for specific physical readouts is left to subsequent work.

Related first-order differential descriptions of Gaussian covariance data have been developed by L\'opez-Sald\'ivar, Man'ko, and Man'ko for unitary and nonunitary Gaussian-state evolution \citep{lopez2020differential}, and by L\'opez-Sald\'ivar for multimode Gaussian equilibrium states and their thermodynamic parametrization \citep{lopez2023differential}. These approaches begin with quantum or thermal evolution and derive differential equations for Gaussian-state parameters. The construction considered here is complementary: it begins with the Bures--Wasserstein metric on covariance space and introduces an independent canonical momentum, producing an inertial cotangent-bundle dynamics whose strongly damped limit is a first-order gradient flow.

For continuous-variable quantum systems, positive definiteness is necessary but not sufficient for a covariance matrix to represent a physical state. In addition, the Robertson--Schr\"odinger condition
\begin{equation}
\Sigma+\frac{i\hbar}{2}\Omega\succeq0,
\label{eq:quantum_uncertainty}
\end{equation}
where $\Omega$ is the standard symplectic form, must hold, equivalently imposing a lower bound on the symplectic eigenvalues of $\Sigma$ \citep{weedbrook2012gaussian,adesso2014}. The ordinary spectral barrier considered here does not in general enforce this symplectic condition, although the stronger bound $\Sigma\succeq(\hbar/2)I$ is sufficient in standard quadrature conventions (in the $[q,p]=2i$ convention used in related work, the vacuum sits at $\Sigma=I$ and the floor at unity; the two differ only by an overall quadrature scale). Similarly, Gaussian separability criteria involve the uncertainty condition after partial transposition and are not determined by the ordinary eigenvalue floor \citep{simon2000peres,werner2001bound}. The present Hamiltonian lift may provide a starting point for covariance dynamics subject to such quantum constraints, but preserving the quantum-admissibility or Gaussian-PPT regions would require barriers or potentials formulated in terms of symplectic eigenvalues or the corresponding uncertainty matrices.

\section*{Acknowledgments}
The author acknowledges the use of AI assistants for structural brainstorming, language refinement, and \LaTeX{} editing during the drafting of this manuscript. The author bears full responsibility for the accuracy and originality of the scientific arguments and equations presented here.

\bibliographystyle{elsarticle-num}
\bibliography{bib,bib_reviewer_additions}
\end{document}